\documentclass[runningheads]{llncs}
\usepackage[T1]{fontenc}
\usepackage{graphicx}
\usepackage{amsmath}
\usepackage{amssymb}
\usepackage{algorithm}
\usepackage{algorithmic}
\usepackage{tikz}
\usepackage{fullpage} 
\newcommand{\network}{G}
\newcommand{\gamedef}{(\network, (b_i)_{i \in [n]}, (d_i)_{i \in [n]}, (c_i)_{i \in [n]})}
\newcommand{\gamedefreduced}{(\network \setminus m, (c_i)_{i \in [n] \setminus \{m\}}, (b_i)_{i \in [n] \setminus \{m\}}, (d_i)_{i \in [n] \setminus \{m\}})}
\newcommand{\support}{ D(\network, (b_i)_{i \in [n]})}

\newtheorem{assumption}{Assumption}
\newtheorem{observation}{Observation}

\begin{document}
\title{Computation of Nash Equilibria of Attack and Defense Games on Networks}
\titlerunning{Nash Equilibria of Attack and Defense Games on Networks}
%
\author{Stanisław Kaźmierowski\inst{1} \and
Marcin Dziubiński\inst{1}\orcidID{0000-0003-1756-2424}}
\authorrunning{S. Kaźmierowski \and M. Dziubiński}

\institute{University of Warsaw, Faculty of Mathematics, Informatics and Mechanics, Banacha 2, 02-097, Warsaw, Poland 
\email{s.kazmierowski@uw.edu.pl}, \email{m.dziubinski@mimuw.edu.pl} }

\maketitle

\begin{abstract}
We consider the computation of a Nash equilibrium in attack and defense games on networks (Bloch et al.~\cite{paper}). We prove that a Nash Equilibrium of the game can be computed in polynomial time with respect to the number of nodes in the network. We propose an algorithm that runs in $O(n^4)$ time with respect to the number of nodes of the network, $n$.

\end{abstract}

\section{Introduction}

International drug trafficking~\cite{mexican_drugs,smuggling_goods}, disrupting the movement of the enemy troops~\cite{military_interdiction_1,military_interdiction_2}, and terrorist attacks~\cite{nuclear_waste} involves strategic actors interacting over a network of bilateral connections. A class of scenarios of this type consists of a network of defenders (e.g. countries connected by common borders) and an attacker attempting to move an undesirable object (e.g. a bomb or a package of drugs) through a network to a targeted node (e.g. a targeted country). Each defender is interested in his security but investment in protection spills over to subsequent defenders on the potential routes of attack.

In a recent paper, Bloch, Chatterjee, and Dutta~\cite{paper} introduce a game theoretic model that captures such scenarios. In the model, an attacker (node $0$) and $n$ defenders are all connected in a fixed network. The attacker chooses a target node and an attack path in the network from her location (node 0) to the location of a targeted node. In the event of a successful attack, the attacker gains the value assigned to the target node. If successfully attacked, the targeted defender loses his value while every other node on the path remains intact. To prevent potential losses,  every defender can invest in costly protection to increase the probability of stopping a potential attack. An attack can be stopped by every defender on the attack path. Bloch~et~al.~\cite{paper} establish the existence of mixed strategy Nash equilibria (NE) in the model and obtain a partial characterization of the NE as well as a full characterization for the networks that form a line. They prove that the set of nodes attacked with positive probability in NE is unique, and under certain redefinition, the model has a unique NE. They provide a set of non-linear equations describing the strategies in a NE when the set of nodes attacked with positive probability is given. Whether this set can be computed efficiently and, consequently, whether a NE of the model can be computed efficiently, was left an open question.

\paragraph{Our Contribution}

We provide an algorithm for calculating a Nash equilibrium (NE) of the model proposed in~\cite{paper}. We prove that the algorithm runs in polynomial time with respect to the number of players. More in detail, we use the idea of reducing the network by removing the nodes that are not attacked under any NE, while maintaining all of the possible paths of attack. We identify a subset of defending nodes called \textit{linkers} which are the subset of nodes that are never attacked in any NE. After removing the linkers from the network, every node can be reached by the attacker by a path of increasing values for the attacker. Using this observation, and building on the idea for computing NE for linear networks, where the nodes are connected in order of ascending values, presented in~\cite{paper}, we obtain a polynomial time algorithm that finds a NE of the model for any connected network.

\paragraph{Related Work}
The problem of strategic transportation and interception of unwanted traffic through a network is known as the problem of network interdiction. In its classic formulation, the problem involves two players: an evader, who sends the traffic through the network, and the interceptor, whose objective is to stop the traffic. The problem of network interdiction has been studied extensively over the past years. In different models, the evader goal can be minimizing the length of the path from the source to the sink~\cite{shortest_path}, avoiding detection by devices placed by the interdictor on the edges~\cite{detection}, or increasing the flow between the source and the sink~\cite{increasing_flow_1,military_interdiction_2,military_interdiction_1}. The applications of network interdiction problem range from smuggling goods~\cite{mexican_drugs,smuggling_goods} and detecting the nuclear material~\cite{nuclear_waste} to disrupting the enemy troops movement~\cite{military_interdiction_1,military_interdiction_2}. See~\cite{survey_1} and~\cite{survey_2} for excellent surveys of this type of model.

The literature on NE computation is vast, and we restrict attention to the closest related works.~\cite{detection} consider a zero-sum game model where the evader chooses any path on the graph and the interdictor chooses an edge in the graph. If this edge is on the path chosen by the evader, then the evader is detected with a fixed probability assigned to this edge. One of the contributions of~\cite{detection} is reducing the linear programming problem which is a classic method for solving the zero-sum games to a polynomial number of constraints and variables when the considered size of strategy set of evader is of a possibly exponential size.

Closer to the model considered in this paper, is the model~\cite{idd} in which each of the $n$ defenders chooses whether to protect himself or not. The defenders are connected in the directed weighted graph with weights of the edges reflecting the probability of being contaminated by the incoming neighbours in case they are successfully attacked. There are two main settings considered, first where the direct attack is a result of a random event (e.g. a pandemic), and second, where the attacks are coordinated by an attacker with an incentive to maximize the number of infected nodes (e.g. hacker attacks). The authors propose a polynomial time algorithm for computing a NE for the models they consider.

\section{The Model} \label{sec:model}

We consider a game, introduced in~\cite{paper}, between an \emph{attacker} (player 0) and $n$ \emph{defenders} (target nodes). We will use $[n] = \{1,2,\ldots,n\}$ to denote the set of defenders. The attacker and the defenders are connected in a network modeled by an undirected graph, $\network = \langle V,E \rangle,$ where $V = [n] \cup \{0\}$ is the set of nodes and $E \subseteq \binom{V}{2}$ is the set of edges ($\binom{V}{2}$ denotes the set of all $2$-element subsets of $V$). We assume that graph $\network$ is connected, meaning that every target is reachable from the attacker by some simple path $p$ (i.e. a path where every node appears at most once).\footnote{Throughout the paper when using the term path we will mean a simple path.} Given a graph $\network$, we will use $E(\network)$ to denote the set of edges in $\network$ and $V(\network)$ to denote the set of nodes in $\network$. 

The attacker attacks a selected defender in the network by reaching him through a path starting at $0$ and ending at the defender. Each defender $j \in [n]$ has a value $b_j > 0$ describing his strategic importance to the attacker. If defender $j$ is successfully attacked, the attacker receives a payoff of $b_j$. 

If attacked successfully, the defender $j$ obtains a negative payoff of $-d_j$. Each defender, anticipating a possible attack, can invest in protection (interception probability). Intercepting the attack means stopping the attacker, regardless of whether the defender is the target or simply lies on an attack path. The investment of defender $j$ increases the probability $x_j \in [0,1]$ of intercepting an attack and comes at a cost, $c_j(x_j)$. The cost of protection is an increasing, differentiable, and strictly convex function of the protection level and the cost of no protection is $0$, i.e. $c_j(0) = 0$. We make the following assumption about the cost functions:

\begin{assumption}\label{ass:marginal_cost}
For every cost function $c_j(x_j)$ of defender $j \in [n]$, we assume that $c_j'(1) \geq d_j$ and $c_j'(0) = 0$.
\end{assumption} 
Assumption~\ref{ass:marginal_cost} implies that the only scenario in which the best response of the defender $j$ might be the ``perfect defense'' (i.e. $x_j=1$) is when he is the only attacked node.

In~\cite{paper} it is assumed that the cost function $c_j(x_j) = x_j^2/2$ for the final presented results, but many important results are proven for a wider class of cost functions satisfying Assumption~\ref{ass:marginal_cost}. 

Defenders choose their interception probabilities independently and simultaneously with the attacker choosing a target $j \in [n]$ and an attack path $p$ from 0 to $j$. Let $P(\network)$ denote the set of all paths in graph $\network$ originating at the attacker node, $0$, and $t(p)$ denote the terminal node of path $p$. Given any path $p \in P(\network)$ and $j \in p$, the set of predecessors of $j$ in $p$ is $Pred(p, j) = \{k \in p : \text{ $k$ lies on path $p$ between 0 and $j$ }\}$.  Fix a vector of interception investments $(x_1 , \ldots , x_n )$. For any node $j$ on path $p$, we let $\alpha_j(p, (x_i)_{i \in [n]})$ denote the probability that the attack along $p$ reaches $j$

\begin{equation*}
    \alpha_j(p, (x_i)_{i \in [n]}) = {\displaystyle \prod_{k \in Pred(p,j)} (1-x_k)}.
\end{equation*}

The probability that the attack on target $j$ along path $p$ is successful is given by

\begin{equation*}
    \beta_j(p, (x_i)_{i \in [n]}) = \alpha_j(p, (x_i)_{i \in [n]}) \cdot (1-x_j).
\end{equation*}

The set of pure strategies of the attacker is defined by the set $P(\network)$ of all paths originating at $0$. The set of pure strategies of every defender $j \in [n]$, the level of protection, is the interval $[0,1]$. Pair $(p, (x_i)_{i \in [n]})$ describes a pure strategy profile, with the payoff of the attacker given by

\begin{equation*}
    U(p, (x_i)_{i \in [n]}) = \beta_{t(p)}(p, (x_i)_{i \in [n]}) b_{t(p)},
\end{equation*}
and the payoff of defender $j$ given by
\begin{equation}
V_j(p, (x_i)_{i \in [n]}) = 
\begin{cases}
 \beta_{j}(p, (x_i)_{i \in [n]}) (-d_{j}) - c_j(x_j) \text{, if } j = t(p),\\
 -c_j(x_j) \text{, otherwise.}
\end{cases}
\end{equation}

We allow the attacker to use mixed strategies, choosing a probability distribution $\pi$ over all paths in $P(\network)$. Let $\Delta(P(\network))$ denote the set of all probability distributions over $P(\network)$. The expected payoff of the attacker from a mixed strategy profile $(\pi, (x_i)_{i \in [n]}) \in \Delta(P(G))\times [0,1]^n$ is given by

\begin{equation} \label{eq:payoff_attacker}
    U(\pi, (x_i)_{i \in [n]}) = \sum_{p \in P(\network)} \pi(p) \beta_{t(p)}(p, (x_i)_{i \in [n]}) b_{t(p)}.
\end{equation}

The expected payoff of defender $j$ is given by 
\begin{equation}\label{eq:payoff_defender}
    V_j(\pi, (x_i)_{i \in [n]}) = \sum_{\substack{p \in P(\network) \\ t(p) = j}} \pi(p) \alpha_j(p, (x_i)_{i \in [n]}) (1 - x_j) (-d_j) - c_j(x_j).
\end{equation}

Following~\cite{paper} we use the following assumption on defenders' importance.
\begin{assumption}\label{ass:b_not_equal}
For any two defenders $i$ and $j$, $b_i \neq b_j$.
\end{assumption}
This assumption means that no two defenders have the same strategic importance to the attacker. Moreover, without the loss of generality, we will assume that the defenders are numbered in increasing order with respect to their strategic importance to the attacker, i.e. $i < j \implies b_i < b_j$. 

\begin{definition} [Attack and defense game on a network] \label{game}
Quadruple $\gamedef$ defines an attack and defense game on a network with network $G$, set of players $V(G)$, defenders' cost functions, $c_j$, attacker's evaluations, $b_j$, and defenders' evaluations $d_j$.
\end{definition}

We are interested in calculating (mixed strategy) Nash equilibria (NE) of attack and defense games on a network defined by Definition 1. A strategy profile $(\pi^*, (x^*_i)_{i \in [n]})$ is a NE if and only if for every mixed strategy $\pi \in \Delta(P(\network))$ of the attacker, $U(\pi^*, (x^*_i)_{i \in [n]}) \geq U(\pi, (x^*_i)_{i \in [n]})$, and for every node $j \in [n]$ and every strategy $x_j \in [0,1]$, $V_j(\pi^*, (x_j, (x^*_i)_{i \in [n] \setminus\{j\}})) \leq V_j(\pi^*, (x^*_i)_{i \in [n]})$.

\section{Properties of the Nash Equilibria}

In this section, we recall important properties of the NE of the attack and defense game on a network. The properties follow from~\cite{paper} and are crucial for the computational results we obtain.

First, Bloch, Chatterjee, and Dutta~\cite{paper} establish the existence of mixed strategy NE in the game.
\begin{theorem}[Bloch et al.~\cite{paper}]\label{th:NE_existence}
The attack and defense game on a network always admits a Nash equilibrium in mixed strategies.
\end{theorem}

Second, they establish sufficient and necessary conditions for the existence of pure strategy NE.
\begin{lemma}[Bloch et al.~\cite{paper}]\label{lem:pureNE}
The described model yields NE in pure strategies if and only if the value $b_n$ of node $n$ satisfies
\begin{equation} \label{eq:pure_1}
    b_n (1 - c_n'^{(-1)}(d_n) ) \geq b_j
\end{equation}
for all $j$ such that there is a path $p$ from 0 to $j$ that does not contain $n$. 
\end{lemma}
Note, that as $c_n$ is a strictly convex, differentiable function, the inverse function $c_n'^{(-1)}$, of its differential, $c_n'$, is well-defined. 

Deciding whether the condition introduced in Lemma~\ref{lem:pureNE} is satisfied can be done in time $O(n)$ by the following straightforward algorithm. After removing the node $n$ from the graph, all the nodes that remain connected to node $0$ by a path form a set of nodes that can be reached by the attacker with a path that does not contain $n$. For this set of nodes, we check whether Inequality~\eqref{eq:pure_1} is satisfied. If the condition is met then every profile where the attacker chooses a path $p$ that terminates at node $n$, defender $n$ chooses investment of $c_n'^{(-1)}(d_n)$ (value obtained from finding the derivative of payoff function of the $n$'th defender and comparing it to $0$) and every other defender chooses investment of $0$ is a pure strategy NE. From now on we will focus on the parameters of the model that do not yield NE in pure strategies.

\subsection{Properties of Mixed Strategies Nash Equilibria}
Given a game $\Gamma = \gamedef$, let 
\begin{equation*}
    D_{\Gamma}(\pi, (x_i)_{i\in[n]}) \subseteq [n],
\end{equation*}
denote the set of all defenders attacked with positive probability under the strategy profile $(\pi, (x_i)_{i\in[n]})$. The following lemma about the independence of set $D_{\Gamma}(\pi, (x_i)_{i\in[n]})$ from a considered strategy profile $(\pi, (x_i)_{i\in[n]})$, that is a NE of $\Gamma$, follows from the proof of Theorem 2~\cite{paper}.

\begin{lemma}[Bloch et al.~\cite{paper}, Theorem 2]\label{lem:unique_set_attacked}
Given the Assumption~\ref{ass:b_not_equal}, for every attack and defense game on network $\Gamma = \gamedef$, and every two strategy profiles $(\pi, (x_i)_{i\in[n]})$ and $(\pi', (x'_i)_{i\in[n]})$ that are NE of $\Gamma$,
\begin{equation*}
 D_{\Gamma}(\pi, (x_i)_{i\in[n]}) = D_{\Gamma}(\pi', (x'_i)_{i\in[n]}).
\end{equation*}
\end{lemma}

By Lemma~\ref{lem:unique_set_attacked}, the set of nodes attacked with positive probability in equilibrium depends only on the game's parameters. Therefore, given a game $\Gamma = \gamedef$ we will denote this set by $D(\Gamma)$. We will call nodes in $D(\Gamma)$ \emph{non-neutral nodes}. From proof of Theorem 2~\cite{paper}, the set of non-neutral nodes is invariant under the vector of values $(d_i)_{i \in [n]}$ and the vector of cost functions $(c_i)_{i \in [n]}$ (as long as they satisfy Assumption 1). This is stated in the following lemma.

\begin{lemma}[Bloch et al.~\cite{paper}]\label{lem:D_invariant}
Let $\Gamma = \gamedef$. For any cost functions vector $(c_i')_{i \in [n]}$, satisfying Assumption~\ref{ass:marginal_cost}, and values vector $(d_i')_{i \in [n]}$ it holds
\begin{equation*}
    D(\Gamma) = D(\Gamma'),
\end{equation*}
where $\Gamma' = (\network, (b_i)_{i \in [n]}, (d_i')_{i \in [n]}, (c_i')_{i \in [n]})$. 
\end{lemma}

Following Lemma~\ref{lem:D_invariant}, for the remaining part of the paper, we will denote the set of nodes attacked in every NE of the game by $\support$. The set of nodes $[n] \setminus \support$ is never attacked under any NE. We call nodes in $[n] \setminus \support$ \emph{neutral nodes}. We have the following observation.

\begin{observation}
Every neutral node $j$ maximizes his payoff in every NE by choosing a strategy $x_j=0$.
\end{observation}

When the network, $G$, and the values of the nodes, $(b_i)_{i \in [n]}$ are clear from the context, we will use $D$ instead of $\support$ to denote the set of non-neutral nodes and $[n]\setminus D$ to denote the set of neutral nodes.

For a non-neutral node, $j$, let $P^j$ denote the set of all paths from $0$ to $j$ chosen by the attacker with positive probability in some NE of the game. Formally, path $p$ from $0$ to $j$ in $G$ belongs to $P^j$ if and only if there exists a strategy profile $(\pi, (x_i)_{i\in[n]})$ that is a NE of the game, such that $\pi(p) > 0$. Bloch et al.~\cite{paper} prove that any two paths in $P^j$  can differ only on the set of neutral nodes. Moreover, non-neutral nodes on any two paths in $P^j$ are aligned in the same sequence from the attacker to the target. This is stated by the following lemma.

\begin{lemma}[Bloch et al.~\cite{paper}]\label{lem:single_path}
For any two paths $p, p'$ in $P^j$,
\begin{equation*}
 Pred(p,j) \cap \support = Pred(p',j) \cap \support.   
\end{equation*}
Moreover, if $k,l \in Pred(p,j) \cap \support$ then
    \begin{equation}
        k \in Pred(p,l) \iff k \in Pred(p',l).
    \end{equation}
\end{lemma}

Following Lemma~\ref{lem:single_path}, for every non-neutral node $j$, we denote the unique sequence of his predecessors from $D$ on any path in $P^j$ by $p^j$. We call $p^j$ \emph{the equilibrium attack path of $j$}. The equilibrium attack paths are not always paths in the original graph, as they can lack some of the neutral nodes that are essential to their connectivity. 

If a non-neutral node, $k \in D$, lies on an equilibrium attack path of another node non-neutral, $j \in D$, his equilibrium attack path, $p^k$, is a subsequence of $p^j$. This is stated by the following lemma.

\begin{lemma}[Bloch et al.~\cite{paper}]\label{lem:path_prefixes}
    Given two non-neutral nodes, $k$ and $j$, if $k$ is an element of $p^j$ then $p^k$ is a subsequence of $p^j$, i.e. that for some $m \in \{2,3, \ldots, |p^j| - 1\}$, $p^k$ is a sequence of first $m$ elements of $p^j$.
\end{lemma}
From Lemma~\ref{lem:path_prefixes}, it follows that the set of nodes $\{0\} \cup D$ and the set of equilibrium attack paths $\{p_j\}_{j \in D}$ constitute a tree that is invariant under the vector of cost functions $(c_i)_{i \in [n]}$ (as long as they satisfy Assumption 1) and the vector of values $(d_i)_{i \in [n]}$. Therefore, for a given game $\gamedef$, we denote this tree by $T(G, (b_i)_{i \in [n]})$ and call it an \emph{equilibrium attack tree}.

The concept of the equilibrium attack tree allows for the following redefinition of the game.

\begin{definition}[Equilibrium attack tree game]\label{def:eq_attack_game}
An equilibrium attack tree game induced by the attack and defense game on network $\gamedef$ is the attack and defense game on network $(T(G, (b_i)_{i \in [n]}),$ $(b_i)_{i \in D},$ $(d_i)_{i \in D},$ $(c_i)_{i \in D})$.
\end{definition}

In such a game, every defender is connected to the attacker by exactly one path -- his equilibrium attack path.
It means that every mixed strategy $\pi$ of the attacker is described by vector $(q_i)_{i \in D}$, which determines the probabilities of attack on every node.

\subsection{NE of the Equilibrium Attack Tree Game}
Given an equilibrium attack tree game $(T(G, (b_i)_{i \in [n]}),$ $(b_i)_{i \in D},$ $(d_i)_{i \in D},$ $(c_i)_{i \in D})$, let $D_0 \subseteq D$ denote the set of all the neighbours of $0$ in tree $T(G, (b_i)_{i \in [n]})$. By~\cite{paper}, the first-order conditions that have to be fulfilled by any NE of the game are

\begin{align}
    &x_j^* = 1 - \frac{U}{b_j} \text{, if } j \in D_0, \label{x_0}\\
    &x_j^* = 1 - \frac{b_{k(j)}}{b_j} \text{, if } j \in D \setminus D_0, \label{x_i} \\
    &q_j^* = \frac{c_j'\left(1 - \frac{U}{b_j}\right)}{d_j}  \text{, if } j \in D_0, \label{q_0}\\
    &q_j^* = \frac{b_{k(j)} \cdot c_j'\left(1 - \frac{b_{k(j)}}{b_j}\right)}{U \cdot d_j} \text{, if } j \in D \setminus D_0, \label{q_i}\\
    &\sum_{j} q_j^* = 1. \label{q_sum}
\end{align}
where $k(j)$ is the direct predecessor of $j$ in the equilibrium attack path $p^j$ and $U$ is the equilibrium utility of the attacker. 

Equations \eqref{x_0} and \eqref{x_i} are obtained from the equations guaranteeing that the attacker is indifferent among the targets in the support.
\begin{alignat*}{2}
& b_j (1 - x_j^* ) = U, &\text{ for } j \in D_0, \\
& b_j (1 - x_j^* ) = b_{k(j)}, &\text{ for } j \notin D_0.
\end{alignat*}

Equations \eqref{q_0} and \eqref{q_i} are obtained from maximizing the payoff function of every defender defined in \eqref{V_i_redefined}. First, we calculate the derivative
\begin{equation}\label{V_i_redefined}
    \frac{\partial V_j(q, x_1, \ldots, x_n)}{\partial x_j} = \alpha_j x_j q_j^* d_j - c'_j(x_j).
\end{equation}
The function $V_j(x_j)$ is concave, therefore it is only increasing in an interval where $\alpha_j x_j q_j^* d_j \geq c'_j(x_j)$. It follows from Assumption~\ref{ass:marginal_cost}, that $0$ is in this interval while $1$ is not, therefore the derivative is equal to 0 in the maximum, hence
\begin{equation} \label{eq:derivative}
    c_j'(x_j^*) = \alpha_j q_j^* d_j.
\end{equation}
In any NE the attacker is indifferent over the strategies in the support, i.e. 
\begin{equation*}
    U = b_j \alpha_j(1-x_j^*).
\end{equation*}
After transforming this equation, we get
\begin{equation*}
    \alpha_j = \frac{U}{b_j(1-x_j^*)}.
\end{equation*}
This means that the equation \eqref{eq:derivative} states
\begin{equation*} 
    c_j'(x_j^*) = \frac{U q_j^* d_j}{b_j(1-x_j^*)},
\end{equation*}
hence
\begin{equation*}
    q_j^* = \frac{c_j'(x_j^*)b_j(1-x_j^*)}{U \cdot d_j}.
\end{equation*}
Using equations \eqref{x_0} and \eqref{x_i} we get \eqref{q_0} and \eqref{q_i}, respectively.
Equation \eqref{q_sum} states that the probabilities in any mixed strategy of the attacker sum up to 1. We conclude this subsection by stating the uniqueness of the solution to the introduced set of equations.

\begin{theorem}[Bloch et al.~\cite{paper}]\label{th:NE_uniqueness}
Given Assumption~\ref{ass:b_not_equal}, the proposed set of first-order conditions~\eqref{x_0}-\eqref{q_sum} yields exactly one solution
\begin{equation*}
    ((q^*)_{i \in D}, (x_i)_{i \in D}, U) \in [0,1]^{|D|} \times [0,1]^{|D|} \times [0,1]
\end{equation*}
that is the unique NE of the equilibrium attack tree game. 
\end{theorem}

\subsection{Properties of the Equilibrium Attack Tree}
In this subsection, we present the properties of the equilibrium attack tree that follows from Theorem~\ref{th:NE_uniqueness}. Consider a non-neutral node $j \in D \setminus D_0$. In the NE of the equilibrium attack tree game $(T(G, (b_i)_{i \in [n]}),$ $(b_i)_{i \in D},$ $(d_i)_{i \in D},$ $(c_i)_{i \in D})$, node $j$ is attacked through the equilibrium attack path $p^j = (0,p^j_1,p^j_2,\ldots,k(j),j)$, and the probability $\alpha_j$ of attacker successfully reaching the node $j$ is
\begin{equation*} 
    \alpha_j = \prod_{i \in \{p^j_1,p^j_2,\ldots,k(j)\}} (1 - x^*_i).
\end{equation*}
Using the \eqref{x_0} and \eqref{x_i}, we get
\begin{multline}\label{eq:alpha_j_minimizing}
\alpha_j = \left(1 - \left(1 - \frac{U}{b_{p^j_1}}\right)\right) \prod_{i \in \{p^j_2,\ldots,k(j)\}} \left(1 - \left(1 - \frac{b_{i-1}}{b_{i}} \right)\right) = \frac{U}{b_{p^j_1}} \prod_{i \in \{p^j_2,\ldots,k(j)\}} \left(\frac{b_{i-1}}{b_{i}} \right) = \frac{U}{b_{k(j)}}.
\end{multline}
The nodes in $D$, that can be a direct predecessor of a node $j$ in his equilibrium attack path, are the non-neutral nodes that can be reached in graph $G$ from $j$ by any path that does not contain any other node from $\{0\}\cup D$. Let $N(j, D, G) \subset D$ denote the set of these nodes. Formally, non-neutral node $i$ is in $N(j, D, G)$ if and only if there exists a path $p$ from $j$ to $i$ in $G$ that does not contain any nodes from $(D \cup \{0\}) \setminus \{i,j\}$.

Equilibrium attack paths are chosen by the attacker to maximize her payoff. Equation~\eqref{eq:alpha_j_minimizing} states, that the smaller the value $b_{k(j)}$, of the direct predecessor of $k(j)$ of node $j \in D \setminus D_0$ on equilibrium attack path $p^j$, the greater the probability of reaching the node $j$ by the attacker. We conclude this with the following observation, which states how the attacker chooses the equilibrium attack tree for a given graph $G$, set of nodes $D_0$ and their evaluations $(b_i)_{i \in D_0}$.

\begin{observation}[Bloch et al. \cite{paper}] \label{obs:connection_smallest}
For any node $j \in D_0$, the attacker maximizes her payoff in the NE of the equilibrium attack tree game by attacking $j$ directly. For any node $j' \in D \setminus D_0$ the attacker maximizes her payoff in the NE of the equilibrium attack tree game by attacking node $j'$ along the equilibrium attack path where the direct predecessor of $j'$ is node $i \in N(j, D, G)$ with the lowest value $b_i$. 
\end{observation}

\section{Computation of Mixed Strategy NE}

The main challenge of computing a NE of a given attack and defense game, $\gamedef$ is computing the set of all non-neutral nodes, $\support$. To tackle this problem, we introduce the idea of network reduction by a subset of neutral nodes. We prove that reducing the network by any subset of neutral nodes retains a particular correspondence between the Nash equilibria of the original and the reduced model (in particular, both games yield the same equilibrium attack tree game). Using network reduction, we show that the equilibrium attack tree of the given game can be found in polynomial time when the set of non-neutral nodes is known. Next, we introduce an important subset of neutral nodes called \emph{linkers}. After reducing the network by the set of linkers, the problem of finding the set of non-neutral nodes is easier. We propose an algorithm that allows for finding the set of non-neutral nodes of a given attack and defense game on a network. The algorithm generalizes the idea of finding the set of non-neutral nodes when the considered network is a linear graph with ascending values $b_j$, presented in~\cite{paper}, to finding this set when the considered network is an arbitrary graph. When the set of non-neutral nodes in the game $\gamedef$ is found, we calculate the NE of the corresponding equilibrium attack tree game. Finally, we show how to reconstruct a NE of a $\gamedef$ from the NE of the corresponding equilibrium attack tree game.

\subsection{Network Reduction}\label{sec:reducing}

For a given game $\gamedef$, its' set of non-neutral nodes $D$, and any neutral node $m \in [n] \setminus D$, let us construct a graph, called \emph{$G$ reduced by $m$}, obtained by removing node $m$ and adding links between all pairs of neighbours of $m$ that are not connected by an edge in $G$. We denote a graph $\network$ reduced by $m$ by $\network \setminus m$. Formally $V(\network \setminus  m ) = V(\network) \setminus \{ m \} $ and $E( \network \setminus m) = E(G) \setminus \{\{i,m\} : i \in V(G) \} \cup \{\{i,k\} : i \neq k \land \{i,m\} \in E(G) \land \{m,k\} \in E(G) \}$.

Let $h_m^G : P(\network) \rightarrow P(\network \setminus m )$ be a function such that, for a given path $p \in P(G)$,
\begin{equation*}
    h_m^G(p) = 
    \begin{cases}
        p &\text { if } m \notin p, \\
        p \setminus \{ m \} &\text{ otherwise.}
    \end{cases}
\end{equation*}
Function $h_m^G$ maps paths emerging from $0$ in graph $G$ to paths emerging from  0 in graph $G \setminus m$. 

Function $h_m^G$, defined for the set of the pure strategies of the attacker, naturally extends to a function $H_m^G : \Delta(P(\network)) \rightarrow \Delta(P(\network \setminus m))$ such that, for every probability distribution $\pi \in \Delta(P(G))$ over the set of paths in $G$,

\begin{equation*}
H_m^G(\pi) = \sum_{p \in P(\network)} \pi(p) \cdot h_m^G(p).
\end{equation*}

\begin{lemma} \label{lem:reduction}
Let $\Gamma$ = $\gamedef$. Node $m \in [n] \setminus D$ is a neutral node and strategy profile $(\pi^*, (x_i)^*_{i \in [n]})$ is a NE of $\Gamma$ if and only if the strategy profile $(H_m^G(\pi^*), (x_i)^*_{i \in [n] \setminus m})$ is a NE of $\Gamma \setminus m = \gamedefreduced$.
\end{lemma} 

\begin{proof}
Notice that the derivative of defender $j \in [n]$ payoff function, $V_j$, is given by
\begin{equation*}
    V_j'(x_j) = d_j \cdot \sum_{\substack{p \in P(\network) \\ t(p) = j}} \pi\left(p\right) \alpha_j\left(p, \left(x\right)_{i \in [n]}\right) - c_j'\left(x_j\right).
\end{equation*} 
Assumption~\ref{ass:marginal_cost} on costs functions guarantees that the maximum of $V_j$ is inside the interval $(0,1)$. As function $V_j$ is concave, we can find this maximum by solving $V_j'(x_j) = 0$. We get
\begin{equation} \label{eq:maximizing_defender_payoff}
    x_j = \left(c_j'\right)^{-1} \left(d_j \cdot \sum_{\substack{p \in P(\network) \\ t(p) = j}} \pi\left(p\right) \alpha_j\left(p, \left(x_i\right)_{i \in [n]}\right)\right).
\end{equation}
In any NE, any defender $j$ chooses the defense investment given by Equation~\eqref{eq:maximizing_defender_payoff} to maximize his payoff.

For the right to left implication, consider a NE $(\pi^*, (x^*_i)_{i \in [n] \setminus \{m\} })$ of a game $\Gamma \setminus m$. We will prove that when node $m$ is neutral, every strategy profile $(\pi, (x_m = 0, (x^*_i)_{i \in [n] \setminus \{m\} }))$ that satisfies $H_m^G(\pi) = \pi^*$ is a NE of the game $\Gamma$. Notice that $x_m = 0$ implies that, for every path $p \in P(G)$,
\begin{equation} \label{eq:alpha_equals}
    \alpha_j(p, (x_i)_{i \in [n]} ) = \alpha_j(h_m^G(p), (x_i)_{i \in [n] \setminus \{m\}}).
\end{equation}
By~\eqref{eq:alpha_equals} every defender $j \in [n] \setminus \{m\}$,
\begin{multline}
    \sum_{\substack{p \in P(\network \setminus \{m\}) \\ t(p) = j}} \pi(p) \alpha_j(p, (x_i)_{i \in [n] \setminus \{m\}}) = \sum_{\substack{p \in P(\network \setminus \{m\}) \\ t(p) = j}} \sum_{\substack{p '\in P(\network) \\ h_m^G(p') = p}} \pi(p') \alpha_j(p', (x_i)_{i \in [n]}) = \\
    \sum_{\substack{p \in P(\network) \\ t(p) = j}} \pi(p) \alpha_j(p, (x_i)_{i \in [n]}).
\end{multline}

As Equation~\eqref{eq:maximizing_defender_payoff} is satisfied for every defender $j \in [n] \setminus \{m\}$ by the strategy profile $(\pi^*, (x^*_i)_{i \in [n] \setminus \{m\} })$,
notice that every defender $j \in [n]$ cannot increase his payoff by deviating from $(\pi, (x_m = 0, (x^*_i)_{i \in [n] \setminus \{m\} }))
$ if the strategies of all the other players remain unchanged. Therefore, the only player that can benefit from changing her strategy in the strategy profile $(\pi, (x_m = 0, (x^*_i)_{i \in [n] \setminus \{m\} }))$ is the attacker. 

Consider any path $p \in P(G)$ such that $\pi(p) = 0$. Notice, that
\begin{multline} \label{eq:attacker_remain}
    U(p, (x_m = 0, (x_i)_{i \in [n] \setminus \{m\}})) = \\ U(h_m^G(p), (x_i)_{i \in [n] \setminus \{m\}})
    \leq  U(\pi^*, (x_i)_{i \in [n] \setminus \{m\}}) = U(\pi, (x_m = 0, (x_i)_{i \in [n] \setminus \{m\}})),
\end{multline}
hence the attacker also cannot increase her payoff by deviating from $(\pi, (x_m = 0, (x^*_i)_{i \in [n] \setminus \{m\} }))$.
The inequality follows from the NE definition and both equalities follow from the Equation~\eqref{eq:alpha_equals}.

The strategy profile  $(\pi, (x_m = 0, (x^*_i)_{i \in [n] \setminus \{m\} }))$ is a NE of the attack and defense game on network $\Gamma \setminus m$, because none of the players can increase their payoff by deviating from it.

The proof of reverse implication is analogous. \qed
\end{proof}

The reduction of the game extends to any set of neutral nodes by iterative reduction of neutral nodes one by one. First, note that for any game $\gamedef$, the corresponding set $D$ of non-neutral nodes and any two neutral nodes $j$, $k \in D \setminus [n]$
\begin{equation}\label{eq:reduction_order}
    H_k^{(G \setminus j)} \circ H_j^G = H_j^{(G \setminus k)} \circ H_k^G.
\end{equation}
The reduction of the game extends to any set $S \subseteq ([n] \setminus D)$ of neutral nodes by iterative reduction of nodes from $S$ one by one. Equation~\eqref{eq:reduction_order} guarantees that reduction by the set of nodes is invariant to the ordering in which we choose nodes from $S$.

\begin{definition}\label{def:set_reduction}
    For a given game $\gamedef$, the corresponding set of non-neutral nodes $D$, any subset of neutral nodes $S \subseteq ([n] \setminus D)$, and any sequence $s = \{s_1, s_2, \ldots,s_{|S|}\}$ of all the nodes in $S$, the reduction of $G$ by $S$ with sequence $s$ is defined as
    \begin{equation*}
        H^G_{S,s} = 
        \begin{cases}
            H^{G \setminus s_1}_{(S \setminus \{s_1\}), (s_2,\ldots,s_{|S|})}, &\text{ if $|S| > 2$,} \\
            H^{G \setminus s_1}_{s_2}, &\text{ if $|S| = 2$}.
        \end{cases}
    \end{equation*}
\end{definition}
As the reduction of the network is independent of the order of nodes from $S$, we denote it with $H^G_S$. Reducing the network by a given neutral node $i$ can be done in $O(n^2)$ time and reducing the network by a given set of nodes, $S\subseteq [n]$, can be done in time $O(|S| \cdot n^2)$.

\subsection{Linkers} \label{linkers} 
We now introduce an important set of nodes called \emph{linkers}. Let us call a node $i$ a \textit{linker} if he is not directly connected to the attacker and all of his neighbours' evaluations, $b_j$, are greater than $b_i$, i.e. $\{0,i\} \notin E(\network)$ and $\left( \{i,j\} \in E\left(\network\right) \implies b_j > b_i \right)$. All linkers are neutral nodes, which we state in the lemma below.

\begin{lemma} \label{lem:linker_is_neutral}
Every linker is a neutral node.
\end{lemma}

\begin{proof}
Consider a linker $m \in [n]$ and any path $p \in P(\network)$ such, that $t(p) = m$, i.e. $m$ is a terminal node of $p$. Let $k(m)$ denote the direct predecessor of node $m$ on path $p$. Notice that the probability $\beta_m(m,p)$ of the successful attack on node $m$ through path $p$ satisfies
\begin{equation*}
    \beta_m(p, (x_i)_{i \in [n]}) = (1 - x_i) \beta_{k(m)}(p \setminus \{m\}, (x_i)_{i \in [n]}) \leq \beta_{k(m)}(p \setminus \{m\}, (x_i)_{i \in [n]}).
\end{equation*}

As $b_m < b_{k(m)}$ from the linker definition, for any strategy $p$, the strategy $p \setminus \{m\}$ yields a strictly greater payoff to the attacker. Therefore, the strategy $p$ is not in the NE support of the attacker. 

No path $p \in P(\network)$ with terminal $t(p) = m$ is in the attacker support in any NE, hence considered node $m$ is a neutral node. \qed
\end{proof}

Let $L(\network) \subseteq [n]$ denote the set of all linkers in graph $\network$. By Lemma~\ref{lem:linker_is_neutral}, the set $L(\network)$ is a subset of the set of all neutral nodes. Following the game reduction by the set of neutral nodes, we can reduce the graph $G$ by $L(\network)$ while retaining the correspondence between the NE of the original and the reduced model. We will call the graph $\network \setminus L(\network)$ a \textit{proper} graph.

The following example illustrates reducing a graph by its linker.
\begin{example}
In the graph shown in Figure~\ref{fig:linkers_1}, node 1 is a linker.
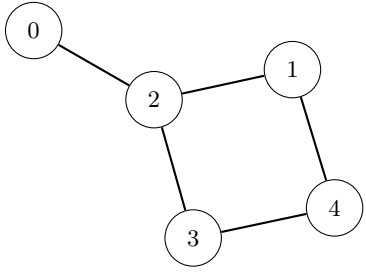
\begin{figure}[ht]
\centering
\begin{tikzpicture}[x=4cm,y=4cm] 
  \tikzset{     
    e4c node/.style={circle,draw,minimum size=0.75cm,inner sep=0}, 
    e4c edge/.style={sloped,above,font=\footnotesize}
  }
  \node[e4c node] (1) at (0.40, 1.53) {2}; 
  \node[e4c node] (2) at (0.53, 1.07) {3}; 
  \node[e4c node] (3) at (0.86, 1.63) {1}; 
  \node[e4c node] (4) at (0.00, 1.76) {0}; 
  \node[e4c node] (5) at (1.00, 1.17) {4};

  \path[draw,thick]
  (3) edge[e4c edge]  (5)
  (1) edge[e4c edge]  (3)
  (1) edge[e4c edge]  (4)
  (2) edge[e4c edge]  (5)
  (2) edge[e4c edge]  (1)
  ;
\end{tikzpicture}
\caption{Graph with a linker (node 1)} 
\label{fig:linkers_1}
\end{figure}  
We can reduce graph $\network$ by node 1, connecting all of 1's neighbours that are not directly connected. As a result, we get the proper graph shown in Figure~\ref{fig:linkers_2}, where nodes 2 and 4 are now directly connected and every node is connected to the attacker with at least one ascending path of indices.
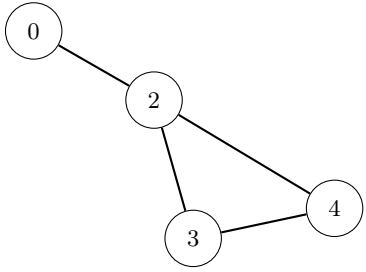
\begin{figure}[ht]
\centering
\begin{tikzpicture}[x=4cm,y=4cm]
  \tikzset{     
    e4c node/.style={circle,draw,minimum size=0.75cm,inner sep=0}, 
    e4c edge/.style={sloped,above,font=\footnotesize}
  }
  \node[e4c node] (1) at (0.40, 1.53) {2}; 
  \node[e4c node] (2) at (0.53, 1.07) {3};  
  \node[e4c node] (4) at (0.00, 1.76) {0}; 
  \node[e4c node] (5) at (1.00, 1.17) {4};

  \path[draw,thick]
  (1) edge[e4c edge]  (4)
  (2) edge[e4c edge]  (5)
  (1) edge[e4c edge]  (5)
  (2) edge[e4c edge]  (1)
  ;
\end{tikzpicture}
\caption{Graph after removing the linker (node 1)} 
\label{fig:linkers_2}
\end{figure}  
\end{example}

The following observation is a direct consequence of the properties of network reduction.
\begin{observation}\label{obs:same_eq_at_tree_game}
Let $\Gamma = \gamedef$, graph $G' = G \setminus L(G)$ is the proper graph of $G$, and $\Gamma' = (\network', (c_i)_{i \in [n] \setminus L(G)}, (b_i)_{i \in [n] \setminus L(G)}, (d_i)_{i \in [n] \setminus L(G)})$. Games $\Gamma$ and $\Gamma'$ yield the same equilibrium attack tree game.
\end{observation}

Determining whether a given node $i \in [n]$ is a linker can be done in time $O(n)$, hence finding the set $L(G)$ of all linkers can be done in time $O(n^2)$. As every node $i$ in the proper graph has a neighbour $j$ of a lower index, the following observation emerges.
\begin{observation} \label{obs:acending_path}
    Every node in a proper graph is connected to the attacker with at least one path of ascending indices. 
\end{observation}
As a consequence of Observation~\ref{obs:acending_path}, we have the following lemma, characterizing the set of non-neutral nodes for proper graphs.

\begin{lemma}
    If graph $\network$ is proper, the set of non-neutral nodes $D$ of the attack and defense game is $\{k,k+1,\ldots,n\}$ for some $k \in [n]$.
\end{lemma}

\begin{proof}
Let us assume that $k$ is the lowest index of a node attacked with positive probability in NE. We will prove that every node with an index greater than $k$ is also attacked with positive probability. Let us assume that node $k+1$ is not attacked in NE. From Observation~\ref{obs:acending_path}, we know there is at least one ascending path of nodes from 0 to $k+1$. 
If there is such a path that does not contain $k$, then $k+1$ has to be attacked. If he were not attacked, then he would not defend himself, and therefore attacking him would yield a greater payoff to the attacker than attacking node $k$. 
If every ascending path from 0 to $k+1$ contains $k$, then again, $k+1$ has to be attacked. If he was not attacked, he would not defend himself, and therefore the attacker could reach him with the same probability as the node $k$, but $k+1$ would yield a greater payoff.
The same reasoning for every other node with an index greater than $k$ shows that if $k$ is the lowest index of a non-neutral node, and graph $\network$ is proper, hence $D = \{k,k+1,\ldots,n\}$. \qed
\end{proof}

\subsection{Computation of the Equilibrium Attack Tree}

Consider game $\gamedef$, with a proper graph $G$, and the corresponding set of non-neutral nodes, $D$. The equilibrium attack tree $T(G, (b_i)_{i \in [n]})$ can be found in polynomial time with the following algorithm.
\begin{algorithm}[H]
 \hspace*{\algorithmicindent} \textbf{Input:} proper graph $G$ and a set $D \subseteq V(G)$ \\
 \hspace*{\algorithmicindent} \textbf{Output:} equilibrium attack tree $T$ 
\begin{algorithmic}[1]
\STATE $G' = G \setminus ([n] \setminus D)$
\STATE $T$ = empty graph
\STATE V($T$) = $\{0\} \cup D$
\FOR{$j \in D$}
    \IF{$(0,j) \in$ E($\network'$)}
        \STATE E($T$) = E($T$) $\cup \{(0,j)\}$
    \ELSE
        \STATE find $N(j, G')$ \COMMENT{the set of neighbours of $j$ in graph $G'$}
        \STATE $i = \min(N(j, G'))$
        \STATE E($T$) = E($T$) $\cup \{\{i,j\}\}$
    \ENDIF
\ENDFOR 
\RETURN $T$
\end{algorithmic}
\caption{Constructing equilibrium attack tree}
\label{alg:equilibrium_attack_tree}
\end{algorithm}
From Observation~\ref{obs:acending_path}, we know that every node in a proper graph has at least one neighbour of a smaller index, hence in every iteration of the \textit{for} loop, a new edge is added to the graph $T$. As the resulting graph $T$ is a connected graph with $n+1$ vertices and $n$ edges, it is in fact a tree.
From Observation~\ref{obs:same_eq_at_tree_game} we know that graphs $G$ and $G'$ yield the same equilibrium attack tree. Observation~\ref{obs:connection_smallest} states that every neighbour of $0$ in $G'$ is directly connected to $0$ in $T$, and every other node is connected in $T$ to his neighbour in $G'$ of the lowest index, which concludes the correctness of Algorithm~\ref{alg:equilibrium_attack_tree} when the input set of nodes is the set of non-neutral nodes.

The dominant procedure when considering the time complexity of Algorithm~\ref{alg:equilibrium_attack_tree} is finding the graph $G'$, which can be done in $O((n - |D|) \cdot n^2$).

\subsection{Finding the Lowest Node Index in $D$}
In this section, we consider game $\gamedef$, with a proper graph $G$, and show how to find the corresponding set of non-neutral nodes, $D$. Let $k^*$ denote the lowest index of a node in $D$. We formulate the condition which is satisfied by $k^*$ alone and test this condition for all the possible values of $k \in [n]$, finding $k^*$.
Using Algorithm~\ref{alg:equilibrium_attack_tree}, for every $k \in [n]$ we can find the equilibrium attack tree $T$ for graph $\network$ assuming that $D = \{k,k+1,\ldots,n\}$. 

For every $k \in [n]$ we define a function $F_k: [0,b_n] \rightarrow \mathbb{R}_{\geq 0}$, such that, for a given payoff $U$ of the attacker,
\begin{equation*}
    F_k(U) = \sum_{i} q_i^*(U,k),
\end{equation*}
where $q_i^*(U,k)$ is given by Equations \eqref{q_0} and \eqref{q_i}, for $D = \{k,k+1,\ldots,n\}$. 
$F_k$ has the following properties.
\begin{enumerate}
    \item It is strictly decreasing in $U$, as every element of the sum is strictly decreasing in $U$.
    \item $F_{k^*}(U^*) = 1$, where $U^*$ denotes the attacker payoff at the equilibrium and $k^*$ denotes the lowest index of a node attacked with positive probability in NE.
\end{enumerate}
The condition on $k^*$ is
\begin{equation}\label{inequality}
    F_{k^*}(b_{k^*}) \leq 1 < F_{k^*}(b_{{k^*}-1}),
\end{equation}
as it implies
\begin{equation} \label{condition}
    b_{k^*} \geq U^* > b_{{k^*}-1}.
\end{equation}

The set of first-order conditions \eqref{x_0}-\eqref{q_i} guarantees that the payoff to the attacker is the same for every pure strategy in the support. The payoff from every pure strategy outside of the support is not greater than $b_{k^*-1}$, hence it is smaller than $U^*$. This means, that the attacker cannot increase her payoff by changing her strategy. Neither can the defenders, as each one of them already maximizes his payoff. Therefore, the strategy profile $((q^*)_{i \in D}, (x_i)_{i \in D})$ defined by Equations \eqref{x_0}-\eqref{q_sum} describe the NE of equilibrium attack tree game  $(T(G, (b_i)_{i \in [n]}),$ $(b_i)_{i \in D},$ $(d_i)_{i \in D},$ $(c_i)_{i \in D})$ (which we know is unique from Theorem~\ref{th:NE_uniqueness}), with $D = \{k^*, k^* + 1, \ldots, n\}$.

\subsection{Calculating the NE of an Equilibrium Attack Tree Game for a Proper Graph}\label{NE}

Consider attack and defense game $\gamedef$, where graph $G$ is proper. To calculate the strategy profile that is the NE of the corresponding equilibrium attack tree game, knowing that $D = \{k^*, k^* + 1, \ldots, n\}$, we need to find the payoff of the attacker $U^*$. This means solving the equation $F_{k^*}(U) = 1$
\begin{equation}\label{U*calculation}
    \sum_{i \in D_0}  \frac{c_i'(1 - \frac{U}{b_i})}{d_i} + \sum_{i \in D \setminus D_0}\frac{b_{k(i)} \cdot c_i'(1 - \frac{b_{k(i)}}{b_i})}{U \cdot d_i} = 1.
\end{equation}

In the case of the cost functions $c_i(x_i)$ being of the form $c_i(x_i) = x_i^2/2$, Equation~(9) takes the form
\begin{equation*}
        \sum_{i \in D_0}  \frac{1 - \frac{U}{b_i}}{d_i} + \sum_{i \in D \setminus D_0}\frac{b_{k(i)} \cdot 1 - \frac{b_{k(i)}}{b_i}}{U \cdot d_i} = 1.
\end{equation*}
This equation can be transformed into a quadratic equation after multiplying both sides by $U$, and it can be solved in linear time with respect to the number of nodes.

After establishing the payoff of the attacker, the last thing to do is to calculate the solutions of equations 
\eqref{x_0}-\eqref{q_i}. Each equation can be solved in a constant time, as we only need to calculate the value of the function $x_i^2/2$ at a given point.

\section{Calculating the Strategies in the NE for any Graph}

Consider $\gamedef$, where graph $G$ is any connected graph. We showed how to find the set $D$ and the NE $((q_i^*)_{i \in D}, (x_i^*)_{i \in D})$ of the attack equilibrium tree game $(T(G, (b_i)_{i \in [n]}),$ $(b_i)_{i \in D},$ $(d_i)_{i \in D},$ $(c_i)_{i \in D})$ by first calculating the proper graph $G' = G \setminus L(G)$, then applying the method of finding the set of non-neutral nodes $D$ and finally calculating the NE, $((q_i^*)_{i \in D}, (x_i^*)_{i \in D})\}$, of the corresponding equilibrium attack tree game. In this section, we show how to retrieve a strategy profile $(\pi, (x_i)_{i \in [n]})$ that is a NE of an attack and defense game on network, $\gamedef$, from the NE $((q_i^*)_{i \in D}, (x_i^*)_{i \in D})$ of the corresponding equilibrium attack tree game.

Let $R(i,j, G, D)$ denote the set of all paths $p$ in $G$ from $i \in D$ to $j \in D$, that do not contain any other node from $D$. To find the set $R(i,j, G, D)$ we remove nodes in $D \setminus \{i,j\}$ from $G$. If $i$ and $j$ are in two different components then $R(i,j, G, D) = \emptyset$. In general, set $R(i,j, G, D)$ can contain (exponentially) many different paths, and therefore can be difficult to find, however, we can obtain any of these paths in time $O(n^2)$. We will denote such a path by $r_{ \{i,j\}}(G,D)$.

Considering an equilibrium attack path, $p^j \in P(T(G, (b_i)_{i \in [n]}))$, we create a path $p^{j,res} \in P(G)$ by replacing the edge between $m \in p^j$ and his predecessor $k(m) \in p^j$ with a path $r_{\{k(m),m\}}(G,D)$. From the reduction procedure, it follows that $r_{\{k(m),m\}}(G,D)$ exists for every such pair of nodes and it can be $\{m,k(m)\}$ if and only if $\{m,k(m)\} \in E(G)$. Let $\pi^* \in \Delta(P(G))$ be
\begin{equation*}
    \pi^*(p) =  
    \begin{cases}
        q_j^* \text{, if } p = p^{j,res}, \\
        0 \text{, otherwise}.
    \end{cases}
\end{equation*}

\begin{observation}
Strategy profile $(\pi^*, ((0)_{i \in [n] \setminus D},$ $(x_i^*)_{i \in D}))$ describes a NE of $\gamedef$.
\end{observation}
This follows from the game reduction by the set of neutral nodes, as by reversing the reduction of $G$ by set $D' = [n] \setminus D$, we can define the mapping $H^G_{D'}$ 
where $(H^G_{D'})^{-1}(p^j) = p^{j,res}$ 
for every node $j \in D$, as for every $j$, and every mapping $H^G_{D'}$, $H^G_{D'}(p^{j,res}) = p^j$. 

The procedure of reconstructing NE of the original game from the $((q_i^*)_{i \in D}, (x_i^*)_{i \in D})$ runs in time $O(|D| \cdot n^2)$, as it requires finding exactly $|D| \leq n$ paths $r_{i,j}(G,D)$. 
\section{Computational Complexity} \label{sec:complexity}

A NE of the attack and defense game $\gamedef$ can be found by the following procedure.
\begin{algorithm}[H]
 \hspace*{\algorithmicindent} \textbf{Input:} attack and defense game $\gamedef$ \\
 \hspace*{\algorithmicindent} \textbf{Output:} a NE of a given game 
\begin{algorithmic}[1]
\STATE find the set $L(\network)$ of all linkers in $G$ 
\STATE calculate the proper graph $\network' = \network \setminus L(\network)$ 
\FOR{$k \in \{1,\ldots,n\}$}
    \STATE calculate $\network_k'$ = $\network' \setminus \{1,\ldots,k-1\}$ 
    \STATE construct equilibrium attack tree $T_k$ for graph $\network_k'$ (Algorithm~\ref{alg:equilibrium_attack_tree})
    \STATE construct function $F_k(U)$ using the equations \eqref{q_0} and \eqref{q_i}
    \IF{inequality \eqref{inequality} holds}
        \STATE save $k^* = k$
        \STATE save $T_{k^*} = T_k$
        \STATE break
    \ENDIF  
\ENDFOR 
\STATE calculate $U^*$ solving \eqref{U*calculation}
\STATE calculate strategies of the attacker and the defenders solving the equations \eqref{x_0}, \eqref{x_i}, \eqref{q_0} and \eqref{q_i}
\STATE reconstruct the NE of the general game
\RETURN NE of the general game
\end{algorithmic}
\caption{Finding the NE of the general model}
\label{alg:solving}
\end{algorithm}
The pessimistic time complexities of all the used procedures were established when the given procedure was introduced.
The dominant operation of Algorithm~\ref{alg:solving} is the calculation of the reduced graph (line 4), which can require time $O(n^3)$ in every iteration of the main loop. Therefore, the pessimistic run time of Algorithm~\ref{alg:equilibrium_attack_tree} is $O(n^4)$.

Algorithm 2 can be easily generalized to compute Nash equilibria of the game for arbitrary cost functions (satisfying Assumptions~\ref{ass:marginal_cost}) other than $c_i(x) = x^2/2$. In the case of such cost functions, the pessimistic time cost of Algorithm~\ref{alg:solving} is polynomial with respect to the number of players as long as Equation~\eqref{U*calculation} can be solved in polynomial time.

\section{Conclusions} \label{sec:conlusions}

In this paper, we proposed a method for finding a NE of attack and defense games on networks~\cite{paper}. The proposed algorithm runs in polynomial time with respect to the number of nodes in the network. 

The idea of reducing the network by the set of linker nodes that results in a proper graph, although simple, allows us to make an important observation that every node in a proper graph can be reached by the attacker by at least one ascending path. This idea can be used for the subclass of attack and interception games on networks, where the defenders make their decisions independently and only the target node is influenced by the attack. However, if the defenders could coordinate, it could happen that a linker node, although never attacked, is protected to defend more valuable nodes that can be reached only through this linker. 

\subsubsection*{Acknowledgements}
This work was supported by the Polish National Science Centre through grant 2018/29/B/ST6/00174. 

\bibliographystyle{splncs04}
\bibliography{paper}

\end{document}